\documentclass{article}

\usepackage[inline]{enumitem}
\usepackage[margin=1in]{geometry}
\usepackage{amsthm}
\usepackage{amsmath}
\usepackage{amssymb}
\usepackage{subcaption}
\usepackage{etoolbox}
\usepackage{booktabs}
\usepackage{ifthen}
\usepackage{listings}
\usepackage{calc}
\usepackage{enumitem}
\usepackage{wrapfig}

\usepackage{tikz}
\usetikzlibrary{shapes}
\usetikzlibrary{arrows}
\usetikzlibrary{patterns}

\newtheorem{theorem}{Theorem}[section]
\newtheorem{lemma}[theorem]{Lemma}
\newtheorem{corollary}[theorem]{Corollary}
\newtheorem{Definition}[theorem]{Definition}

\lstdefinelanguage{Pseudo}{
  keywords={let, and, or, not, while, loop, for, if, else, end, do, then, return},
  morecomment=[l]{//},
}

\newcounter{caseCounter}
\newcounter{subCaseCounter}

\newenvironment{distinction}[1][reset]{
  \newcommand{\case}[1]{
    \stepcounter{caseCounter}
    \item[{
      \hspace*{-2em}
      \normalfont\color{black}Case~\arabic{caseCounter}. \bf\boldmath ##1.
    }]
  }
  \ifstrequal{#1}{reset}{\setcounter{caseCounter}{0}}{}
  \begin{description}[leftmargin=\widthof{\hspace*{-2em}Case~1.\quad}]
}{
  \end{description}
}

\newenvironment{subdistinction}{
  \newcommand{\subcase}[1]{
  \item[\normalfont\unboldmath\color{black}
    Case~\arabic{caseCounter}.\arabic{subCaseCounter}\stepcounter{subCaseCounter}. ##1.]
  }
  \setcounter{subCaseCounter}{1}
  \begin{description}[leftmargin=0em]
}{
  \end{description}
}

\newcommand{\bs}{\backslash}
\newcommand{\jobs}{\mathcal J}
\newcommand{\machines}{\mathcal M}
\newcommand{\tree}{\mathcal T}
\newcommand{\OPT}{\mathrm{OPT}}

\newcommand{\Bgmin}[1]{B_{#1}^{\min }}
\newcommand{\Bg}[1]{B_{#1}}
\newcommand{\Sm}[2] {
  \ifthenelse{\equal{#2}{}}{
    S_{#1}(\blraw{}{})
  } {
    S_{#1}(\blraw{}{#2})
  }
}

\newcommand{\blocker}[1] {
  \mathcal B_{#1}
}

\newcommand{\blsymbol}[1]{%
\mathrm{%
\ifstrequal{#1}{s_}{SA}{%
\ifstrequal{#1}{bs_}{BA}{%
\ifstrequal{#1}{mm_}{BB}{%
\ifstrequal{#1}{m_}{BL}{%
\ifstrequal{#1}{}{}{{\errmessage{Unrecognized blocker: #1}}}}}}}%
}
}

\newcommand{\bltext}[1] {%
\ifstrequal{#1}{s_}{small-to-any}{%
\ifstrequal{#1}{bs_}{big-to-any}{%
\ifstrequal{#1}{mm_}{big-to-big}{%
\ifstrequal{#1}{m_}{big-to-least}{%
  \errmessage{Unrecognized blocker: #1}
}}}}%
}

\newcommand{\Bltext}[1] {%
\ifstrequal{#1}{s_}{Small-to-any}{%
\ifstrequal{#1}{bs_}{Big-to-any}{%
\ifstrequal{#1}{mm_}{Big-to-big}{%
\ifstrequal{#1}{m_}{Big-to-least}{%
  \errmessage{Unrecognized blocker: #1}
}}}}%
}

\newcommand{\bl}[2] {
  \mathcal M(\blraw{#1}{#2})
}

\newcommand{\blraw}[2] {
  \ifthenelse{\equal{#1}{}}{
  \ifthenelse{\equal{#2}{}}{
    \tree
  } {
    \tree^{(#2)}
  }
  }{
  \ifthenelse{\equal{#2}{}}{
    \tree_{\blsymbol{#1}}
  } {
    \tree_{\blsymbol{#1}}^{(#2)}
  }
  }
}

\newcommand{\act}[2] {
  \ifthenelse{\equal{#2}{}}{
    \mathcal A_{#1}(\blraw{}{})
  }{
    \mathcal A_{#1}(\blraw{}{#2})
  }
}

\newcommand{\jnew}{
  j_\mathrm{new}
}
\begin{document}

\title{On the Configuration-LP of the Restricted Assignment Problem%
\thanks{Research was supported by German Research Foundation (DFG) project JA 612/15-1}}

\author{
  Klaus Jansen\\
  \and Lars Rohwedder \\[.5em]
  \and
  Department of Computer Science, University of Kiel, 24118 Kiel, Germany \\
  \{kj, lro\}@informatik.uni-kiel.de
\date{}
}

\maketitle

\begin{abstract} We consider the classical problem
of {\sc Scheduling on Unrelated Machines}. In this problem a set of jobs 
is to be distributed among a set of machines
and the maximum load (makespan) is to be minimized.
The processing time $p_{ij}$ of a job $j$ depends on the machine $i$ it is assigned to.
Lenstra, Shmoys and Tardos gave a polynomial time $2$-approximation for this problem~\cite{Lenstra:1990:AAS:81018.81019}. 
In this paper we focus on a prominent special case, 
the {\sc Restricted Assignment} problem, in which $p_{ij}\in\{p_j,\infty\}$.
The configuration-LP is a linear programming relaxation for the
{\sc Restricted Assignment} problem.
It was shown by Svensson that the multiplicative gap
between integral and fractional solution, the integrality gap,
is at most $2 - 1/17 \approx 1.9412$~\cite{DBLP:journals/siamcomp/Svensson12}.
In this paper we significantly simplify his proof and achieve a bound of $2 - 1/6 \approx 1.8333$. 
As a direct consequence this provides a polynomial $(2 - 1/6 + \epsilon)$-estimation algorithm for
the {\sc Restricted Assignment} problem by approximating the configuration-LP.
The best lower bound known for the integrality gap is $1.5$ and no
estimation algorithm with a guarantee better than $1.5$ exists unless $\mathrm{P} = \mathrm{NP}$.

\end{abstract}

\section{Introduction}%
In the problem {\sc Scheduling on Unrelated Machines} we are given a set of jobs $\jobs$ and 
a set of machines $\machines$.
Each job has to be assigned to exactly one machine.
A solution is a function $\sigma: \mathcal J\rightarrow \mathcal M$. 
Every Job $j\in\mathcal J$ has a processing time $p_{ij}$, that depends on the machine 
$i\in\mathcal M$ it is assigned to.
The objective is to minimize the highest load among all machines, that is
$\max_{i\in\mathcal M}\sum_{j\in\sigma^{-1}(i)} p_{ij}$.
This value is called the makespan of $\sigma$.
The best polynomial algorithm known for {\sc Scheduling on Unrelated Machine} has an approximation guarantee of
$2$~\cite{Lenstra:1990:AAS:81018.81019}. In the same paper it was also shown that
approximating the problem with a rate better than $1.5$ is $\mathrm{NP}$-hard.
Williamson and Shmoys~\cite{DBLP:books/daglib/0030297} as well as Schuurman and Woeginger~\cite{schuurman1999polynomial} asked as an important open question,
whether this can be improved. 
In both cases, the {\sc Restricted Assignment} problem is mentioned as well.
The {\sc Restricted Assignment} problem is a natural special case.
There we have $p_{ij}\in\{p_j,\infty\}$ for a processing time $p_j$ which depends only on the job.
The intuition is that each job is allowed only on a subset of machines
$\Gamma(j) = \{i\in\mathcal M : p_{ij} < \infty\}$, but the processing time on any of
these machines is identical.
The known bounds on approximability of the {\sc Restricted Assignment} match the general problem.

An important instrument for the {\sc Restricted Assignment} problem
is a linear programming formulation called the configuration-LP.
Svensson has shown that this relaxation has an integrality gap of no more than
$2 - 1/17$~\cite{DBLP:journals/siamcomp/Svensson12}.
Since the linear program can be solved up to an error of an arbitrarily small $\epsilon$ in polynomial time,
this gives us a polynomial $(2 - 1/17+\epsilon)$-estimation algorithm.
A $c$-estimation algorithm stands for an algorithm that computes a value $E$
with $\OPT \le E \le c \cdot \OPT$, where $\OPT$ denotes the optimum.
In contrast, a $c$-approximation algorithm would also need to produce a schedule that has
a makespan in this range. It is already $\mathrm{NP}$-hard to estimate the {\sc Restricted Assignment}
problem with a rate better than $1.5$ (as can be shown using the reduction from~\cite{Lenstra:1990:AAS:81018.81019}).

Further progress has been made on the special case of two different processing times, where an upper bound of $2 - 1/3$ was
shown for the integrality gap by Land, Maack, and the first author~\cite{DBLP:conf/swat/JansenLM16} and a $(2 - \delta)$-approximation for a very small $\delta$ by Chakrabarty, Khanna, and Li~\cite{DBLP:conf/soda/ChakrabartyKL15}.

The method used to show the bound on the integrality gap is a local search algorithm that
produces a solution of the mentioned quality. It is, however, not known, whether
the algorithm terminates in polynomial time.
An indication of the potential of this research is given by the closely related 
{\sc Restricted Max-Min Fair Allocation} problem.
This problem maximizes the minimum load instead of minimizing the maximum.
Similar techniques have been employed for this problem and the respective local search algorithm has been 
improved first to a quasi-polynomial running time by Pol{\'{a}}cek and Svensson~\cite{DBLP:journals/talg/PolacekS16}
and eventually to a polynomial one by Annamalai, Kalaitzis, and Svensson~\cite{DBLP:conf/soda/AnnamalaiKS15}.

Another noteworthy branch of research regarding the {\sc Restricted Assignment} problem is the special case
of {\sc Graph Balancing}. In this problem each job is allowed on at most two machines, that is for all $j\in\jobs$ we have $|\Gamma(j)| \le 2$. Using linear programming techniques a $1.75$-approximation was
achieved by Ebenlendr, Krc{\'{a}}l, and Sgall~\cite{DBLP:journals/algorithmica/EbenlendrKS14}.
Those techniques, however, differ very much from those related to the
configuration-LP. Recently Huang and Ott were able to show some further results on
instances of {\sc Restricted Assignment} where the constraint above applies only to big jobs~\cite{lightgb}.

\paragraph{Our contribution.}
As a direction for future research Svensson stated
in his paper~\cite{DBLP:journals/siamcomp/Svensson12}:
\begin{quote}
  {\it
    \noindent
  \lq\lq [To improve the upper bound on the integrality gap] one possibility would be to find a more elegant generalization of the techniques, presented [\dots] for two job sizes, to arbitrary processing times (instead of the exhaustive case distinction presented in this paper).\rq\rq} 
\end{quote}
This accurately captures where our contribution lies.
We improve the upper bound on the integrality gap by coping
with arbitrary many job sizes in a different way.
We also argue that our proof is simpler than the original one (see Section~\ref{differences}). 
\begin{theorem}
  The configuration-LP for the {\sc Restricted Assignment} problem has an integrality gap of at most $2 - 1/6$.
\end{theorem}
\begin{corollary}
  For every $\epsilon > 0$ there exists a polynomial $(2 - 1/6 + \epsilon)$-estimation algorithm for the {\sc Restricted Assignment} problem.
\end{corollary}
\subsection{Differences to Svensson's algorithm}
\label{differences}
We distinguish the jobs sizes only into small and big, whereas the original proof
additionally has a fine distinction of big jobs into medium, large, and huge.

One of the difficulties that arise in Svensson's paper 
when dealing with many job sizes is the 
role of the medium jobs.
There are situations where we want to move huge jobs to certain machines,
but these machines contain too many medium jobs.
In the process of removing those medium jobs,
it is hard to prevent bigger medium jobs from being moved to this machine.
The original paper solves this by rounding the medium jobs to the same size. 
To be precise, the rounding affects only the values in the construction of a dual LP solution, 
which is used in the analysis.
In contrast, we accept that during this process new medium
jobs are moved to the machine, as long as they are bigger than the smallest one on the machine.
Our observation is that when repeating this
process the size of the smallest medium job steadily increases; thus
the process eventually terminates.

Rounding of medium jobs introduces some unpleasant problems.
Our approach avoids these problems making the analysis a lot less
technical. It also appears to be a bottleneck for the solution quality
in the original local search algorithm and with our methods the approximation ratio can be improved
significantly.

\subsection{Notation}
  For jobs we usually use the variable $j$ or variants like $j_B, j'$, etc.
  For machines we use $i$ and the like.
  Since we frequently sum over sets of jobs, for a $H\subseteq \jobs$
  we write $\sum_{j\in H} p_j$ simply as $p(H)$. 
  With other variables indexed by jobs we do the same.

\subsection{The configuration-LP}
A configuration for a machine is a set of jobs that would not exceed the target makespan $T$, if
the jobs were assigned to that machine.
The set of all configurations for a machine $i$ is therefore defined as
\begin{equation*}
  \mathcal C_i(T) = \{C\subseteq \mathcal J : \sum_{j\in C} p_{ij} \le T\} \, .
\end{equation*}
Note that by this definition a job $j$ will only appear in configurations of machines in $\Gamma(j)$.
The idea for the configuration-LP (Figure~\ref{fig:clp}) is that for each machine a mixture of
configurations is selected. A variable $x_{i,C}$ describes which fraction of the configuration
$C$ is used on machine $i$. The first constraint guarantees that the sum of these fractions
on each machine is at most $1$. The second constraint ensures that each job appears in fractional
configurations with a total value of at least $1$.
It is easy to see that an integral solution to the configuration-LP directly correlates with
a solution for the respective instance of the {\sc Restricted Assignment} problem, when $T$ equals
the optimum.
The minimal $T$ for which the LP is feasible is considered the optimum of the relaxation and
is denoted by $\OPT^*$.
\begin{figure}
  \centering
  \begin{subfigure}{.47\textwidth}
  \centering
    \fbox{
      \addtolength{\linewidth}{-2\fboxsep}%
      \addtolength{\linewidth}{-2\fboxrule}%
      \begin{minipage}{\linewidth}
    \begin{align*}
      \sum_{C\in \mathcal C_i(T)} x_{i, C} &\le 1  &\forall i\in \machines \\
      \sum_{i\in \machines}\sum_{C\in \mathcal C_i(T) : j\in C} x_{i, C} &\ge 1 &\forall j\in \mathcal J \\
      x_{i,C} &\ge 0 \\
    \end{align*}
    \end{minipage}}
    \caption{Primal}
    \label{fig:primal}
  \end{subfigure}
  ~
  \begin{subfigure}{.47\textwidth}
    \fbox{
      \addtolength{\linewidth}{-2\fboxsep}%
      \addtolength{\linewidth}{-2\fboxrule}%
      \begin{minipage}{\linewidth}
    \begin{align*}
      \min &\sum_{i\in \machines} y_i - \sum_{j\in \jobs} z_j \\
      s.t.& \\
      y_i &\ge \sum_{j\in C} z_j \quad \forall i\in \machines, C\in \mathcal C_i(T) \\
      y_i, z_j &\ge 0
    \end{align*}
    \end{minipage}}
    \caption{Dual}
    \label{fig:dual}
  \end{subfigure}
  \caption{The configuration-LP}
  \label{fig:clp}
\end{figure}
Although the configuration-LP has an exponential number of variables,  
it can be approximated in polynomial time to a factor of $1 + \epsilon$
for any $\epsilon > 0$ ~\cite{DBLP:conf/stoc/BansalS06}.
A well-known result from linear programming theory is that an unbounded dual implies
an infeasible primal. This will be a central argument in our analysis.

We note that in the more general {\sc Scheduling on Unrelated Machine} problem the configuration-LP has an integrality gap that cannot be bounded by a constant 
lower than 2~\cite{DBLP:journals/scheduling/VerschaeW14}, which indicates that these techniques
are probably not particularly helpful there.
In the restricted variant, however, no instance is known to have an integrality gap higher
than $1.5$. An instance with integrality gap $1.5$ is given in~\cite{DBLP:conf/swat/JansenLM16}.

\subsection{Reduced problem} %
Let us first simplify our problem.
The makespan we aim for is $(1 + R) \OPT^*$, where $R = 5/6$. 
For convenience we scale each processing time by $1/\OPT^*$;
thereby we establish $\OPT^* = 1$.
Unless otherwise stated, statements about configurations or the
configuration-LP are meant in respect to a makespan of $1$.

\begin{Definition}[Small and big jobs]{\rm%
A job $j$ is \emph{small} if $p_j \le 1/2$ and \emph{big} otherwise.
The set of small (big) jobs is denoted by $\jobs_S$
(respectively, $\jobs_B$).}
\end{Definition}
A crucial difference between small and big jobs is that there can be more than one small job in a configuration,
whereas at most one big job may occur. 
\begin{Definition}[{[Valid]} partial schedule]{\rm%
We call a function $\sigma: \jobs\rightarrow \machines\cup\{\bot\}$ a partial schedule
when for each job $j$ either $\sigma(j) = \bot$ or $\sigma(j)\in\Gamma(j)$.
The value $\bot$ is used to indicate that a job has not been assigned, yet.
A partial schedule is said to be \emph{valid} if
for each machine $i\in\machines$ the sum of the processing times of the jobs assigned to it
is at most $1 + R$, that is $p(\sigma^{-1}(i)) \le 1 + R$.}
\end{Definition}
To find a schedule for all jobs, we maintain a valid partial schedule
and extend it one job at a time (in an arbitrary order).
The problem that we will focus on in this paper is given below.
\\[1em]
\fbox{
  \begin{minipage}{.95\linewidth}
  \begin{description}
    \item[\normalfont\color{black}\textit{Input:}] 
      An instance of {\sc Restricted Assignment} with $\OPT^* = 1$,
      a valid partial schedule $\sigma$, and
      a job $\jnew$ with $\sigma(\jnew) = \bot$.
    \item[\normalfont\color{black}\textit{Output:}]
      A valid partial schedule $\sigma'$ with $\sigma'(\jnew) \neq \bot$
      and if $\sigma(j) \neq \bot$ then
      $\sigma'(j) \neq \bot$ for all $j\in\jobs$.
  \end{description}
  \end{minipage}%
}
\\[1em]
Without loss of generality assume that the job identifiers are natural numbers, i.e., $\jobs = \{1, 2, 3, \dotsc\}$,
and the jobs are ordered by size ($p_1 \le p_2 \le \dotsc$).
This provides a total order on the jobs and simplifies the notation in algorithm and analysis.

\section{Local search algorithm}
The key idea for the algorithm is a structure of so-called blockers. 
We motivate it with the thought below. 
Suppose adding $\jnew$ exceeds the capacity of every machine in $\Gamma(\jnew)$.
This means we have to move away some jobs from those machines first. Indeed,
these jobs now have a similar role to $\jnew$.
In other words, we try to find a new machine for them as well.
We can repeat this with the hope that at some point we can move one of the jobs.
To keep track of all dependencies between machines and jobs we use blockers,
which represent moves that are placed on hold to be performed at a later time.

\begin{Definition}[{[Valid]} moves]{\rm%
  A move $(j,\, i)$ (with respect to a partial schedule $\sigma$)
  is composed of a job $j$ and a machine $i$ requiring 
  that $j$ is allowed on $i$, but is assigned to a different machine,
  that is $i\in\Gamma(j)\backslash\{\sigma(j)\}$.
  We call a move $(j, i)$ \emph{valid} if $p(\sigma^{-1}(i)) + p_j \le 1 + R$.}
\end{Definition}
Note that for a valid move $(j, i)$ altering the assignment of $j$ to $i$ does not compromise
the validity of the partial schedule.
Here and subsequently, $\Bg{i}$ denotes the set of big jobs on machine $i$, i.e., $\sigma^{-1}(i) \cap \jobs_B$.
We define
\begin{equation*}
  \Bgmin{i} := \begin{cases}
    \{\min B_i\} &\text{ if } \Bg{i}\neq\emptyset \text{ and} \\
    \emptyset &\text{ otherwise.}
  \end{cases}
\end{equation*}
Recall that jobs are represented by natural numbers and are ordered by size.
As opposed to the rather natural set $\{j\in B_i \,\vert\, \forall j'\in B_i (p_j\le p_{j'})\}$,
$\Bgmin{i}$ always has exactly one or zero elements. 
Also, $\Bgmin{i}$ is always defined, whereas $\min B_i$ is undefined for $B_i = \emptyset$.
Note that $S_i$, which we will introduce later as well, is not analogous to $\Bg{i}$.

\begin{Definition}[Blockers]{\rm%
A blocker is a tuple $(j, i, \Theta)$ where $(j, i)$ is a move and $\Theta$ is the type of
the blocker.
There are four types with the following abbreviations:\\
\begin{description*}
  \item[$(\blsymbol{s_})$] \emph{\bltext{s_} blockers},
  \item[$(\blsymbol{bs_})$] \emph{\bltext{bs_} blockers},
  \item[$(\blsymbol{mm_})$] \emph{\bltext{mm_} blockers},
  \item[$(\blsymbol{m_})$] \emph{\bltext{m_} blockers}.
\end{description*}}
\end{Definition}
As the first part of the name suggests, each \bltext{s_} blocker $(j, i, \blsymbol{s_})$
corresponds to a small job $j$; all other types correspond to big jobs.
The latter part of the type's name is a hint at which jobs are marked undesirable
on the machine corresponding to the blocker.
The algorithm will try to move
jobs that are undesirable on a machine away from it;
at the same time it will not attempt to move such jobs onto the machine.
On machines of small-/big-to-any blockers all jobs are undesirable;
on machines of \bltext{mm_} blockers big jobs are undesirable;
finally on machines of \bltext{m_} blockers only particular big jobs are undesirable.
The goal for this type is to steadily increase the size of the smallest big job. 
To achieve this we would like to remove the smallest job from $B_i$ and at the same time ensure that no smaller jobs
are being added to $B_i$ (\bltext{m_} blockers will ensure that $B_i\neq\emptyset$).
More precisely, all big jobs of index smaller or equal to $\min B_i$ are marked undesirable.

The blocker tree $\tree$ is a set of blockers.
We will elaborate the tree analogy in a moment.
Roughly speaking, the blockers correspond to moves that the algorithm considers
useful for assigning $\jnew$.
For simplicity of notation, we say a move $(j, i)$ is in $\tree$, if there is a blocker $(j, i, \Theta)$ for
some type $\Theta$ with $(j, i, \Theta)\in\tree$.
The blockers corresponding to the specific types are written as
$\blraw{s_}{}$, $\blraw{bs_}{}$, $\blraw{mm_}{}$, and $\blraw{m_}{}$.

With $\tree$ we associate a tree structure, in which the blockers and
one additional root form the nodes.
Each blocker $\blocker{} = (j, i, \Theta)$ has a parent, that is determined by $j$. It is
a blocker for machine $\sigma(j)$ with a type, for which $j$ is regarded undesirable, or
the root if $j = \jnew$.
If there are several candidates for a parent, the algorithm uses the one that was added to the blocker
tree first. We say this blocker activates the job (see also definition of active jobs below).
It is important to understand that the tree is always connected, i.e., a parent exists for every blocker.
This, however, becomes obvious once we complete the description the algorithm.

We denote the individual blockers by $\blocker{1},\dotsc,\blocker{\ell}$
in the order they were added.
For $k \le \ell$ we write $\blraw{}{\le k}$ for the blockers $\blocker{1},\dotsc,\blocker{k}$.
In the final algorithm whenever we remove a blocker $\blocker{k}$, we also remove
all blockers added to the tree after it, that is $\blocker{k+1},\dotsc,\blocker{\ell}$.

From the blocker tree, we derive the machine set $\bl{}{}$ which
consists of all machines corresponding to blockers in $\blraw{}{}$.
The same notation is used for subsets of $\tree$, such as for instance
in $\bl{}{\le k}$ and $\bl{bs_}{}$.

The same machine can appear more than once in the blocker tree.
In that case, the undesirable jobs are the union of the undesirable jobs from all types.
Also, the same job can appear multiple times in different blockers.

\begin{Definition}[Active jobs]{\rm%
  We consider the small jobs $j$ that are undesirable on all other machines 
  they are allowed on, that is
  \begin{equation*}
    \Gamma(j)\backslash\{\sigma(j)\} \subseteq \machines(\blraw{s_}{} \cup \blraw{bs_}{}) .
  \end{equation*}
  We write $\Sm{}{}$ for these jobs.
  They are small jobs that we cannot hope to move at the moment.
  By $\Sm{i}{}$ we denote for a machines $i$ the set of
  jobs in $\Sm{}{}$ that are assigned to $i$, that is $\Sm{i}{} := \Sm{}{}\cap\sigma^{-1}(i)$.
  The set of \emph{active} jobs $\act{}{}$ includes $\jnew$, $\Sm{}{}$ as well as
  all those jobs, that are undesirable on the machine they are currently assigned to.
  Like with $\Sm{i}{}$, we denote by $\act{i}{}$ the set $\act{}{}\cap\sigma^{-1}(i)$.}
\end{Definition}

\subsection{Detailed description of the algorithm}
The general outline of the algorithm is that
as long as $\jnew$ has not been assigned, 
the algorithm (see Algorithm~\ref{algo:simple-extend}) performs a valid move
in the blocker tree, if possible, and otherwise adds a new blocker.
\begin{figure}
  \centering
  \begin{minipage}{.67\textwidth}
\begin{lstlisting}[
  caption=Local search,
  language=Pseudo,
  frame=tb,
  label=algo:simple-extend,
  captionpos=t,
  abovecaptionskip=-\medskipamount,
  keywordstyle=\bfseries,
  mathescape=true,
%  commentstyle=\color{gray!50!black}\textit,
%  escapeinside={(*}{*)},
  basicstyle=\small
]
// Input: Job $\jnew$ and partial schedule $\sigma$
initialize empty blocker tree $\tree$;
loop
  if a move $(j,\, i)$ in $\tree$ is valid then
    Let $\blocker{k}$ be the blocker that activated $j$;
    // Update the schedule
    $\sigma(j) \gets i$;
    if $j = \jnew$ then
      return $\sigma$;
    end
    // Delete $\blocker{k}, \blocker{k+1},\dotsc$
    $\tree\gets\blraw{}{\le k - 1}$;
  else
    choose a potential move $(j, i)$ with minimum value;
    add $(j, i)$ to $\tree$ with the correct type;
  end
end
\end{lstlisting}
\end{minipage}
\end{figure}

\paragraph{Adding blockers.}
The move, that is added to $\tree$, has to meet certain requirements. 
A move, that does, is called a potential move.
For each type of blocker we define a type of potential move:
Potential \bltext{bs_} moves, potential \bltext{mm_} moves, etc.
When a potential move is added to the blocker tree,
its type will then be used for the blocker.
Let $j\in\act{}{}$.
For a move $(j, i)$ to be a potential move of a certain type,
it has to meet the following requirements.
\begin{enumerate}
  \item $(j, i)$ is not already in $\tree$.
  \item The size of $j$ corresponds to the type, i.e.,
        if $j$ is small, $(j, i)$ must be a \bltext{s_} move and
        if $j$ is big, $(j, i)$ must be a \bltext{bs_}, \bltext{mm_}, or \bltext{m_} move.
  \item $j$ is not undesirable on $i$.
        In other words,
        $i\notin \machines(\blraw{s_}{}\cup\blraw{bs_}{})$ and
        if $j$ is big, then $i\notin\bl{mm_}{}$
           and either $i\notin\bl{m_}{}$ or $\min B_i < j$.
   \item The load on $i$ meets certain conditions depending on the type.
     For the sake of readability, these conditions are arranged in table form
     (see Table~\ref{ta:simple-blockers}).\label{en:conditions}
\end{enumerate}
\begin{table*}
  \centering
  \caption{Types of potential moves / blockers w.r.t. a move $(j,i)$}
  \label{ta:simple-blockers}
  \begin{tabular}{lrll}
    \toprule
    Type & Conditions & $\mathrm{val}(j, i)$ & Undesirable \\\midrule
    Small-to-any ($\blsymbol{s_}$) & None & $(1, |\sigma^{-1}(i)|)$ & All jobs \\
    Big-to-any ($\blsymbol{bs_}$) & $p(\Sm{i}{} \cup \Bg{i}) + p_j \le 1 + R$ & $(2, |\sigma^{-1}(i)|)$ & All jobs \\
    Big-to-least ($\blsymbol{m_}$) & $p(\Sm{i}{} \cup \Bgmin{i}) + p_j > 1 + R$ & $(3, - \min B_i)$ & Big jobs $j_B$ \\
                                 & $p(\Sm{i}{}) + p_j \le 1 + R$ & & with $j_B \le \min B_i$ \\
    Big-to-big ($\blsymbol{mm_}$) & $p(\Sm{i}{} \cup B_i) + p_j > 1 + R$ & $(4, |B_i|)$ & Big jobs \\
                                & $p(\Sm{i}{} \cup \Bgmin{i}) + p_j \le 1 + R$ & & \\\bottomrule
  \end{tabular}
\end{table*}
If we compare the conditions in the table, we notice that
for moves of big jobs, there is exactly one type that applies,
unless $p(\Sm{i}{}) + p_j > 1 + R$, in which case the move is
never a potential move.
The table also lists a value for each type of move.
The algorithm will add the move that has the lowest value (in lexicographical order).
It may be useful at this point to take a closer look at the
value of a \bltext{m_} move, that is $(3, -\min \Bg{i})$.
Here $\min \Bg{i}$ is a certain job. The definition only makes sense, because
we use the convention that jobs are identified by natural number $\{1,\dotsc,|\jobs|\}$.

\paragraph{Performing valid moves.}
If a move $(j, i)$ in $\tree$ is valid, it can be performed.
Let $\blocker{k}$ be the blocker that activated $j$.
The algorithm will simply reassign $j$ to $i$ and discard $\blocker{k},\blocker{k+1},\dotsc$.
The idea is that during the next iteration of the loop,
$\blocker{k}$ could be added back to $\tree$, potentially with a different type. 

\subsection{Example}%
\begin{figure}[b!]
\centering
\begin{subfigure}{\linewidth}
  \centering
\begin{tikzpicture}[scale=0.7]
  \draw[pattern=north west lines, pattern color=gray] (0, 0) rectangle (1, 2) node[fill=white, pos=.5] {$\small\jnew$};
  \draw[] (0, 0) rectangle (1, 2);
  \draw[-triangle 60] (1, 1) -- (2.5, 1) node[above, pos=.5] {$\blsymbol{m_}$};
  \draw[dashed] (2.5, 0) rectangle (3.5, 3);
  \draw[pattern=north west lines, pattern color=gray] (2.5, 0) rectangle (3.5, 2);
  \draw (2.5, 2) rectangle (3.5, 2.2);
  \draw (2.5, 2.2) rectangle (3.5, 2.3);
  \draw (2.5, 2.3) rectangle (3.5, 2.6);
  \node at (3, -0.5) {$\blocker 1 (M_1)$};

  \draw[-triangle 60] (3.5, 1) -- (5, 1) node[above, pos=.5] {$\blsymbol{bs_}$};
  \draw[dashed] (5, 0) rectangle (6, 3);
  \draw[pattern=north west lines, pattern color=gray] (5, 0) rectangle (6, 0.8);
  \draw[pattern=north west lines, pattern color=gray] (5, 0.8) rectangle (6, 1.3);
  \draw[pattern=north west lines, pattern color=gray] (5, 1.3) rectangle (6, 1.4);
  \draw[fill=gray] (5, 1.4) rectangle (6, 1.8);
  \node at (5.5, -0.5) {$\blocker 2 (M_2)$};

  \draw[-triangle 60] (6, 1) -- (7.5, 1) node[above, pos=.5] {$\blsymbol{s_}$};
  \draw[dashed] (7.5, 0) rectangle (8.5, 3);
  \draw[pattern=north west lines, pattern color=gray] (7.5, 0) rectangle (8.5, .5);
  \draw[pattern=north west lines, pattern color=gray] (7.5, .5) rectangle (8.5, .7);
  \draw[pattern=north west lines, pattern color=gray] (7.5, .7) rectangle (8.5, 1.6);
  \draw[pattern=north west lines, pattern color=gray] (7.5, 1.6) rectangle (8.5, 2);
  \draw[pattern=north west lines, pattern color=gray] (7.5, 2) rectangle (8.5, 3);
  \node at (8, -0.5) {$\blocker 3 (M_3)$};

  \draw[-triangle 60] (6, 1) .. controls (8, 5) .. (10, 1);
  \node at (9.25, 3.75) {$\blsymbol{s_}$};
  \draw[dashed] (10, 0) rectangle (11, 3);
  \draw (10, 0) rectangle (11, 2);
  \draw[pattern=north west lines, pattern color=gray] (10, 2) rectangle (11, 2.2);
  \draw[pattern=north west lines, pattern color=gray] (10, 2.2) rectangle (11, 2.3);
  \draw[pattern=north west lines, pattern color=gray] (10, 2.3) rectangle (11, 2.6);
  \node at (10.5, -0.5) {$\blocker 4 (M_1)$};
\end{tikzpicture}
\caption{The blocker tree before the move,}
\end{subfigure}
\\[1em]
\begin{subfigure}{.2\linewidth}
\centering
\begin{tikzpicture}[scale=0.7]
  \draw[pattern=north west lines, pattern color=gray] (0, 0) rectangle (1, 2) node[fill=white, pos=.5] {$\jnew$};
  \draw[] (0, 0) rectangle (1, 2);
  \draw[-triangle 60] (1, 1) -- (2.5, 1) node[above, pos=.5] {$\blsymbol{m_}$};
  \draw[dashed] (2.5, 0) rectangle (3.5, 3);
  \draw[pattern=north west lines, pattern color=gray] (2.5, 0) rectangle (3.5, 2);
  \draw (2.5, 2) rectangle (3.5, 2.2);
  \draw (2.5, 2.2) rectangle (3.5, 2.3);
  \draw (2.5, 2.3) rectangle (3.5, 2.6);
  \draw[fill=gray] (2.5, 2.6) rectangle (3.5, 3);
  \node at (3, -0.5) {$\blocker 1 (M_1)$};

\end{tikzpicture}
\caption{after the move, and}
\end{subfigure}
~
\begin{subfigure}{.3\linewidth}
  \centering
\begin{tikzpicture}[scale=0.7]
  \draw[pattern=north west lines, pattern color=gray] (0, 0) rectangle (1, 2) node[fill=white, pos=.5] {$\jnew$};
  \draw[] (0, 0) rectangle (1, 2);
  \draw[-triangle 60] (1, 1) -- (2.5, 1) node[above, pos=.5] {$\blsymbol{m_}$};
  \draw[dashed] (2.5, 0) rectangle (3.5, 3);
  \draw[pattern=north west lines, pattern color=gray] (2.5, 0) rectangle (3.5, 2);
  \draw (2.5, 2) rectangle (3.5, 2.2);
  \draw (2.5, 2.2) rectangle (3.5, 2.3);
  \draw (2.5, 2.3) rectangle (3.5, 2.6);
  \draw[fill=gray] (2.5, 2.6) rectangle (3.5, 3);
  \node at (3, -0.5) {$\blocker 1 (M_1)$};

  \draw[-triangle 60] (3.5, 1) -- (5, 1) node[above, pos=.5] {$\blsymbol{bs_}$};
  \draw[dashed] (5, 0) rectangle (6, 3);
  \draw[pattern=north west lines, pattern color=gray] (5, 0) rectangle (6, 0.8);
  \draw[pattern=north west lines, pattern color=gray] (5, 0.8) rectangle (6, 1.3);
  \draw[pattern=north west lines, pattern color=gray] (5, 1.3) rectangle (6, 1.4);
  \node at (5.5, -0.5) {$\blocker 2 (M_2)$};

\end{tikzpicture}
\caption{after the blocker is added back.}
\end{subfigure}
\end{figure}%
This example demonstrates the process of a move being performed.
In the illustrations the dashed boxes show the capacity of a machine,
hatched jobs are those, which the blocker activates, and the gray job is the one
that will be moved.
The initial state is given in~(a). In the current blocker tree 3 different machines are involved.
The machine corresponding to $\blocker 1$ is the same as the one for $\blocker 4$.
It was first added as a \bltext{m_} blocker and then later as a \bltext{s_} blocker.
The move for $\blocker 4$ is valid. The job was activated by $\blocker 2$.
After the job is moved (see~(b)), 
the blockers starting with the one that activated the gray job
are removed from the blocker tree.
Since the move for $\blocker 2$ is still a potential move,
it can be added back to the blocker tree in the next iteration (see~(c)).
We notice that the value of the move, that is $(2, |\sigma^{-1}(M_2)|)$, has decreased
between (a) and (c).

\section{Analysis}
The correctness of the algorithm relies on two theorems.
First, we need to understand that at any point the blocker tree either contains a valid move,
which can be performed, or there is a potential move that can be added to the blocker tree.
This means that the algorithm does not get stuck at any point during its execution.
Next, we show that after a bounded number of steps the algorithm terminates. 

\begin{lemma}[Invariants]
  \label{lm:inv_}
  At the start of each iteration of the loop,
  \begin{enumerate}
    \item for all $(j, i, \blsymbol{bs_}) = \blocker{k+1}\in\blraw{bs_}{}$:
      $p(\Sm{i}{\le k} \cup B_i) + p_j \le 1 + R$; \label{inv:bs_}
    \item for all $(j, i, \blsymbol{m_}) = \blocker{k+1}\in\blraw{m_}{}$:
      $p(\Sm{i}{\le k} \cup \Bgmin{i}) + p_j > 1 + R$ and \\
      $p(\Sm{i}{\le k}) + p_j \le 1 + R$; \label{inv:m_}
    \item for all $(j, i, \blsymbol{mm_}) = \blocker{k+1}\in\blraw{mm_}{}$:
      $p(\Sm{i}{\le k} \cup B_i) + p_j > 1 + R$ and \\
      $p(\Sm{i}{\le k} \cup \Bgmin{i}) + p_j \le 1 + R$. \label{inv:mm_}
  \end{enumerate}
\end{lemma}
\begin{proof}
During the lifespan of a blocker $\blocker{k+1} = (j, i, \Theta)$, $\blraw{}{\le k}$ does not change; otherwise
$\blocker{k+1}$ would be deleted. By definition, we cannot add potential moves
for jobs in $\Sm{i}{\le k}$ as long as $\blraw{}{\le k}$ does not change. 
At the same time no blockers for small jobs in $\sigma^{-1}(i)$ can already be in $\blraw{}{\le k}$,
since otherwise
$i\in\machines(\blraw{s_}{\le k}\cup\blraw{bs_}{\le k})$ and $\blocker{k+1}$ could not have been added.
For this reason we can deduce that $p(\Sm{i}{\le k})$ is constant during the lifespan of $\blocker{k+1}$.
When a job, which was activated by $\blocker{k+1}$ or an earlier blocker, is removed from $i$,
then $\blocker{k+1}$ is deleted.
For \bltext{bs_} or \bltext{mm_} blockers, all jobs in $B_i$ are activated within $\blraw{}{\le k+1}$
and therefore do not change during the lifespan of $\blocker{k+1}$.
Likewise, for \bltext{m_} blockers, $\Bgmin{i}$ does not change.
We conclude that the invariants are correct.
\end{proof}
\begin{theorem}
  \label{th-potential}
  At the start of each iteration of the loop, 
  there is a valid move in the blocker tree or a potential move of 
  a job in $\act{}{}$.
\end{theorem}
\begin{proof}
  Suppose toward contradiction no potential move of a job in $\act{}{}$ 
  remains and no move in the blocker tree is valid. 
  We will show that the dual of the configuration-LP is unbounded (regarding makespan $1$),
  which implies the primal is infeasible. For that purpose, define the solution $z^*$, $y^*$ below for the dual.
  For all jobs $j$ and machines $i$ let
  \begin{align*}
    z^*_j = \begin{cases}
      \min\{p_j, \frac 5 6 \} &\text{if } j\in\act{}{}, \\
      0 &\text{otherwise;}
    \end{cases}
    \quad
    y^*_i = \begin{cases}
      z^*(\act{i}{}) + \frac 1 6 &\text{if } i\in\bl{bs_}{}, \\
      z^*(\act{i}{}) - \frac 1 6 &\text{if } i\in\bl{s_}{}, \\
      z^*(\act{i}{}) &\text{otherwise.}
    \end{cases}
  \end{align*}
  Note that since $\bl{bs_}{}\cap\bl{s_}{} = \emptyset$, $y^*$ is well-defined:
  On a blocker of either type, all jobs are undesirable, that is to say as long as one of such
  blockers remains in the blocker tree, the algorithm will not add another blocker with
  the same machine. 
  \newtheorem{claim}{Claim}[theorem]
  \begin{claim}
    \label{claim:simple-unboundedness}
    The objective value of the solution is negative, that is $\sum_{j\in\mathcal J} z^*_j > \sum_{i\in\mathcal M} y^*_i$.
  \end{claim}
  \begin{claim}
    \label{claim:simple-feasibility}
    The solution is feasible, that is to say $z^*(C) \le y^*_i$ for all
    $i\in\mathcal M$, $C\in\mathcal C_i$.
  \end{claim}
  It is easy to see that if the claims hold, then they also hold for
  a scaled solution $(\alpha \cdot y^*, \, \alpha \cdot z^*)$ with $\alpha > 0$.
  We can use this to reach an arbitrarily low objective value, thus proving that the dual 
  is unbounded, which means the assumption in the beginning must be false.
\end{proof}

\begin{proof}[Proof of Claim~\ref{claim:simple-unboundedness}]
  The critical part of the prove is that $|\bl{bs_}{}|\le|\bl{s_}{}|$.
  Recall that moves of small jobs have a higher priority than moves of big jobs.
  This means that after a \bltext{bs_} blocker is added to the blocker tree,
  all available moves of small jobs will be added,
  before another \bltext{bs_} blocker can be.
  It is enough to show that whenever a \bltext{bs_} blocker is added,
  there is at least one small job $j_S$ on the corresponding machine and there exists a machine
  $i'\in\Gamma(j_S)\bs\{\sigma(j_S)\}$ that is not in $\machines(\blraw{s_}{}\cup\blraw{bs_}{})$. 
  When this small job is moved, the \bltext{bs_} blocker
  will be deleted. This means as long as the \bltext{bs_} blocker exists,
  there will be a \bltext{s_} blocker added after it,
  but before the next \bltext{bs_} blocker.

  Let $j_B$ be a big job and $i$ some machine. 
  We consider the situation where a \bltext{bs_} blocker for $(j_B,\, i)$ 
  is added. Assume w.l.o.g. that $(j_B, i)$ is not valid. If it is valid, the
  blocker will be removed again instantly.
  At the time $(j_B,\, i)$ is added as 
  a \bltext{bs_} blocker,
  we have that $p(\Sm{i}{} + \Bg{i}) + p_{j_B} \le 1 + R$. 
  On the other hand, since $(j_B,\, i)$ is not valid, we also know that
  $p(\sigma^{-1}(i)) + p_{j_B} > 1 + R$.
  This implies there has to be a small job $j_S$ on $i$ that is not in $\Sm{}{}$.
  By definition of $\Sm{}{}$ there must be
  a machine $i'\in\Gamma(j_S)\backslash\{i\}$ that is not in $\machines(\blraw{s_}{}\cup\blraw{bs_}{})$.
  Thus
  \begin{align*}
    \sum_{j\in\mathcal J} z^*_j \ge z^*_{\jnew} + \sum_{i\in\mathcal M}z^*(\sigma^{-1}(i))
                                > 0 + \sum_{i\in\machines} y^*_i + \frac 1 6 |\bl{s_}{}| - \frac 1 6 |\bl{bs_}{}|
                                \ge \sum_{i\in\machines} y^*_i ,
  \end{align*}
  which concludes the proof.
\end{proof}

\begin{proof}[Proof of Claim~\ref{claim:simple-feasibility}]
  Let $i\in\machines$ and $C\in\mathcal C_i$.
  First, we consider the cases where $i\in\machines(\blraw{s_}{}\cup\blraw{bs_}{})$. Here we will show that
  $y^*_i \ge 1$, which implies the claim, since $z^*(C) \le 1$.
    \begin{distinction}
\case{$i\in\bl{bs_}{}$}
      Then there must be a \bltext{bs_} move $(j_B, i)$ which is not valid;
      that is to say, $p(\sigma^{-1}(i)) + p_{j_B} > 1 + R$.
      If every job $j\in\sigma^{-1}(i)$ has $p_j \le 5/6$, then $p(\sigma^{-1}(i)) = z^*(\sigma^{-1}(i))$ and thus
        \begin{equation*}
          y^*_i = p(\sigma^{-1}(i)) + \frac 1 6
          > 1 + R - p_{j_B} + \frac 1 6 \ge 1 .
        \end{equation*}
        Otherwise let $j'_B\in B_i$ with $p_{j'_B} > 5/6$. Then
        \begin{equation*}
          y^*_i = z^*(\sigma^{-1}(i)) + \frac 1 6 \ge z^*_{j'_B} + \frac 1 6 = 1 .
        \end{equation*}
\case{$i\in\bl{s_}{}$} Then there is a small move $(j_S, i)$, which is not valid,
        i.e., $p(\sigma^{-1}(i)) + p_{j_S} > 1 + R$.
        If there is more than one big job assigned to $i$, we easily see that
        $z^*(\sigma^{-1}(i)) \ge 1$.
        Otherwise at most one job can be rounded down, which implies
        $z^*(\sigma^{-1}(i)) \ge p(\sigma^{-1}(i)) - 1/6$. Hence
        \begin{equation*}
          y^*_i \ge p(\sigma^{-1}(i)) - \frac 2 6 \\
          > 1 + R - p_{j_S} - \frac 2 6 \ge 1 .
        \end{equation*}
    \end{distinction}
  In the following, we study the cases where $i\notin\machines(\blraw{s_}{}\cup\blraw{bs_}{})$. Here it is
  sufficient to show that if there is a $j\in C$ big with $\sigma(j)\neq i$,
  then $z^*_j \le z^*(\act{i}{} \bs C)$.
  The reason is that all small jobs in $C\cap\act{}{}$ are in $\act{i}{}$ or else there
  would be a potential move.
  This means, if $C$ does not contain a big job that are not also in $\sigma^{-1}(i)$,
  then
  \begin{equation*}
    y^*_i = z^*(\sigma^{-1}(i)) \ge z^*(\act{i}{} \cap \jobs_S) + z^*(B_i) \ge z^*(C) .
  \end{equation*}
  Otherwise there is exactly one big $j_B\in C$ and $\sigma(j_B) \neq i$.
  Therefore
  \begin{align*}
    z^*(C) = z^*_{j_B} + z^*(\act{i}{} \cap C)
    \le z^*(\act{i}{}\bs C) + z^*(\act{i}{} \cap C)
    = y^*_i . 
  \end{align*}
  For this purpose, in the cases below let $j\in C$ big with $\sigma(j) \neq i$.
  \begin{distinction}[noreset]
\case{$i\in\bl{mm_}{}\bs\machines(\blraw{s_}{}\cup\blraw{bs_}{})$}
  We observe that Invariant~\ref{inv:mm_} implies that $B_i \neq \Bgmin{i}$ and
  thereby that there are at least two jobs $j_B$, $j'_B$ in $B_i$.
  Thus
  \begin{equation*}
    z^*_j \le 5/6 < z^*_{j_B} + z^*_{j'_B} \le z^*(\act{i}{}\backslash C) .
  \end{equation*}
\case{$i\in\bl{m_}{}\bs\machines(\blraw{s_}{}\cup\blraw{bs_}{}\cup\blraw{mm_}{})$}
  By Invariant~\ref{inv:m_} obviously $B_i$ is not empty.
  Let $j_B = \min B_i$. 
  \begin{subdistinction}
    \subcase{$j > j_B$ and $p_{j_B} \le 5/6$}
      We notice that $j$ is not undesirable on $i$.
      If $(j, i)$ is not a blocker in $\tree$, then it must not be a potential move,
      that is
      \begin{equation}
        \begin{aligned} 
          p(\Sm{i}{} \cup \{j_B\}) + p_j \ge p(\Sm{i}{}) + p_j
                                         > 1 + R .
        \end{aligned}\label{eq:not-potential}
      \end{equation}
      If it is a blocker, then by Invariant~\ref{inv:m_}, inequality~(\ref{eq:not-potential}) holds as well,
      since the blocker must be a \bltext{m_} blocker. Thus
      \begin{align*}
        z^*(\act{i}{}\backslash C) \ge p(\Sm{i}{}\backslash C) + p_{j_B}
                                   \ge p(\Sm{i}{}) - (1 - p_j) + p_{j_B}
                                   > R \ge z^*_j .
      \end{align*}
    \subcase{$j < j_B$ or $p_{j_B} \ge 5/6$}
      Then $p_j \le p_{j_B}$ or $z^*_{j_B} = 5/6 = z^*_{j}$; in particular
      \begin{equation*}
        z^*_j \le z^*_{j_B} \le z^*(\act{i}{}\bs C).
      \end{equation*}
  \end{subdistinction}
\case{$i\notin\bl{}{}$}
  Then since $(j, i)$ is not a potential move, we find $p(\Sm{i}{}) + p_j > 1 + R$.
  Therefore
  \begin{align*}
    z^*_j \le 1 + R - p_j - (1 - p_j)
    < p(\Sm{i}{}) - (1 - p_j)
    \le z^*(\act{i}{}) - z^*(\act{i}{} \cap C)
    = z^*(\act{i}{} \bs C) .
  \end{align*}
    \end{distinction}
  This completes the proof.
\end{proof}

\begin{theorem}
  The algorithm terminates.
\end{theorem}
\begin{proof}
  We will define what we call the signature vector as
  \begin{equation*}
    (\mathrm{val}(\blocker{1}),\mathrm{val}(\blocker{2}), \dotsc,\mathrm{val}(\blocker{\ell}), \infty),
  \end{equation*}
  where $\mathrm{val}(\blocker{k})$ is the value of the move corresponding to
  a blocker $\blocker{k}$ at the time it was added to the tree.
  We argue that the signature vector decreases lexicographically after finitely
  many iterations of the algorithm's loop.
  Since we can bound the number of possible vectors fairly easily,
  we can also bound the overall running time:
  Each move can only appear once in the blocker tree,
  therefore $\ell\le |\jobs|\cdot|\machines|$.
  The range of $\mathrm{val}$ is also finite; hence the number of signature vectors is
  finite as well.

  When the algorithm adds a blocker, clearly the signature vector decreases
  lexicographically. Now consider a valid move that is performed.
  It may be that many consecutive moves are performed, but at some point the algorithm will
  have to add a blocker again.
  We will observe only the last move $(j_0, i_0)$ that is performed before a blocker is added.
  Let $\blocker{k} = (j, i, \Theta)$ be the blocker that activated $j_0$. In particular $j_0$
  was assigned to $i$. 
  We will show that the blocker that will be added in the next iteration
  has a lower value than $\blocker{k}$ had,
  thereby ultimately decreasing the lexicographic order.
  A candidate for this blocker is $(j, i)$. In the case distinction below we show that $(j, i)$
  is a potential move and
  $\mathrm{val}(j, i)$ is lower than it was the last time the blocker was added. This means
  the next blocker that is added must definitely have a lower value.
  The blocker tree, that is $\blraw{}{\le k - 1}$, is identical to the tree when $(j, i)$ was added
  the last time. This means $j$ is still not undesirable.
  \begin{distinction}
    \case{$\blocker{k}$ was a big-/small-to-any blocker} 
      Since all jobs were undesirable on $i$, no jobs have been moved there.
      If $\blocker{k}$ was a \bltext{bs_} blocker, then before $j_0$ was moved to $i_0$,
      $p(\Sm{i}{\le k-1} \cup B_i) + p_j \le 1 + R$ (Invariant~\ref{inv:bs_})
      and this move cannot have increased the left-hand side.
      In other words, $(j, i)$ is still a potential move of the same type.
      Since $j_0$ was removed from $i$, $|\sigma^{-1}(i)|$ has decreased and so has
      $\mathrm{val}(j, i)$.
    \case{$\blocker{k}$ was a \bltext{m_} blocker}
      Because $j_0$ was activated by $\blocker{k}$, it must have been
      the smallest big job on $i$. 
      By Invariant~\ref{inv:m_} we know that before $j_0$ was moved,
      \begin{align}
        p(\Sm{i}{\le k-1}) + p_j &\le 1 + R \text{ and} \label{eq:term-m1} \\
        p(\Sm{i}{\le k-1} \cup \Bgmin{i}) + p_j &> 1 + R. \label{eq:term-m2}
      \end{align}
      (\ref{eq:term-m1}) still holds, since the move has no effect on it.
      If $j_0$ was the only big job, then $B_i = \emptyset$ and $p(\Sm{i}{} \cup B_i) + p_j \le 1 + R$;
      $(j, i)$ is now a \bltext{bs_} move, which has a lower value.
      If $j_0$ was not the only big job on $i$, then $\min B_i$ has increased,
      as no big jobs $j_B$ with $j_B < j_0$ were moved to $i$ since $\blocker{k}$ was added;
      the left-hand side of (\ref{eq:term-m2}) has not decreased
      and $(j, i)$ is still a \bltext{m_} move, but has a lower value as well.
    \case{$\blocker{k}$ was a \bltext{mm_} blocker}
      Since $j_0$ was undesirable regarding $\blocker{k}$, it must be a big job.
      No big jobs were moved to $i$ and $j_0$ was removed; hence $|B_i|$ has
      decreased.
      Before $j_0$ was moved, 
      \begin{align*}
        p(\Sm{i}{\le k-1}) + p_j \le p(\Sm{i}{\le k-1} \cup \Bgmin{i}) + p_j
                                 \le 1 + R,
      \end{align*}
      where we use Invariant~\ref{inv:mm_} in the last inequality.
      The move obviously has no effect on the left-hand side and the inequality must still hold.
     It may, however, be that $(j, i)$ is now a potential move of different type,
      but all other types have lower values; $\mathrm{val}(j, i)$ decreases in any case.
  \end{distinction}
\end{proof}

\section{Conclusion}
We have almost tripled the difference between integrality gap and the classical
approximability result of $2$.
This also marks the gap between best estimation result and best approximation result
(except for an arbitrary small $\epsilon > 0$).
\paragraph{On better bounds.}
The largest gap we have witnessed in an instance is $9/6$ (given in ~\cite{DBLP:conf/swat/JansenLM16}).
The bound on the integrality gap that was proven in this paper, $11/6$,
is not necessarily tight.

An approach for improving the upper bound is to first consider a set of more
constrained instances.
A natural candidate for this could be the set of instances where for
every job $j$ either $p_j = 1$ or $p_j \le 1/2$ (as before,
$1$ is the optimum of the configuration-LP). 
Even there it does not appear to be straight-forward to improve the algorithm,
as we will show below.
With the techniques by Svensson, it was already easy to see that these instances
admit an integrality gap of no more than $11/6$. In this sense,
we believe that a significant bottleneck that was introduced by the medium jobs
in the original algorithm was eliminated by our improvement.

We will show that using certain basic assumptions on the proof 
we can construct a situation where the bound $11/6$ is the best we can get.
In particular, the proof cannot be improved only by altering the constants.
Let $1 + R$ be the makespan we try to achieve.
First, we claim that we need to choose the value $z^*_{j_B}$
for big jobs $j_B$ to be at most $R$:
Let $j_B$ be a big job with $p_{j_B} = 1$.
We consider a situation where $i\in\Gamma(j_B)$ is a machine with $p(\Sm{i}{}) = R + \epsilon$
for some $\epsilon > 0$ and no
jobs besides $\Sm{i}{}$ are assigned to $i$.
The move $(j_B, i)$ must not be a potential move, since
Claim~\ref{claim:simple-unboundedness} relies on the existence of
a small job on $i$ that is not in $\Sm{}{}$ for every non-valid potential \bltext{bs_} move.
Let $C$ be the configuration with only $j_B$ in it. 
Then for feasibility (Claim~\ref{claim:simple-feasibility}) we need that
\begin{equation*}
  z^*_{j_B} = z^*(C) \le y^*_i = z^*(\sigma^{-1}(i)) = p(\Sm{i}{}) = R + \epsilon .
\end{equation*}
Now let there be two machines $i\in\bl{s_}{}$ and $i'\in\bl{bs_}{}$.
$i$ was added because of a move $(j_S, i)$ of a small job $j_S$ with $p_{j_S} = 1/2$
and it is assigned a big job $j_B$ with $p_{j_B} = 1$.
Then $p(\sigma^{-1}(i))$ can be as low as $1 + R - 1/2 + \epsilon$ for any $\epsilon > 0$ without $(j_S, i)$ being valid. We assume that this is the case.
Similarly we assume that $p(\sigma^{-1}(i')) = R + \epsilon$ and the big move to $i'$ is not valid.
In Claim~\ref{claim:simple-feasibility} we rely on $y^*_i\ge 1$ and $y^*_{i'} \ge 1$
and in Claim~\ref{claim:simple-unboundedness} we need that $y^*_i + y^*_{i'} \le z^*(\sigma^{-1}(i)) + z^*(\sigma^{-1}(i'))$.
Together we get
\begin{align*}
  2 \le y^*_i + y^*_{i'}
    \le z^*(\sigma^{-1}(i)) + z^*(\sigma^{-1}(i'))
    = (\frac 1 2 + R  + \epsilon - p_{j_B} + z^*_{j_B}) + (R + \epsilon)
    \le 3 R - \frac 1 2 + 3 \epsilon ,
\end{align*}
which is equivalent to $R \ge 5/6 - \epsilon$.

\paragraph{Beyond estimations.}
The perhaps most interesting
question regarding this problem is, whether a polynomial approximation of the same guarantee
(or at least of guarantee $2 - \delta$ for some constant $\delta > 0$) exists.
Indeed, there are very few problems for which the best known estimation rate is lower than the
best known approximation rate.
Though it seems likely, we cannot know for sure that such an approximation algorithm exists.
Feige and Jozeph even gave a proof that under some complexity assumptions there must be problems
with estimation algorithms superior to the best approximation algorithms~\cite{DBLP:conf/innovations/FeigeJ15}.

On the other hand, for the problem of {\sc Restricted Max-Min Fair Allocation}
efficient variants of the respective local search algorithm were
discovered~\cite{DBLP:journals/talg/PolacekS16,DBLP:conf/soda/AnnamalaiKS15}.
We conjecture that using similar approaches this is possible for the 
{\sc Restricted Assignment} problem as well.

\bibliographystyle{abbrv}
\bibliography{gap_plain}

\begin{thebibliography}{10}

\bibitem{DBLP:conf/soda/AnnamalaiKS15}
C.~Annamalai, C.~Kalaitzis, and O.~Svensson.
\newblock Combinatorial algorithm for restricted max-min fair allocation.
\newblock In {\em Proceedings of the Twenty-Sixth Annual {ACM-SIAM} Symposium
  on Discrete Algorithms, {SODA} 2015, San Diego, CA, USA, January 4-6, 2015},
  pages 1357--1372, 2015.

\bibitem{DBLP:conf/stoc/BansalS06}
N.~Bansal and M.~Sviridenko.
\newblock The santa claus problem.
\newblock In {\em Proceedings of the 38th Annual {ACM} Symposium on Theory of
  Computing, Seattle, WA, USA, May 21-23, 2006}, pages 31--40, 2006.

\bibitem{DBLP:conf/soda/ChakrabartyKL15}
D.~Chakrabarty, S.~Khanna, and S.~Li.
\newblock On $(1, \epsilon)$-restricted assignment makespan minimization.
\newblock In {\em Proceedings of the Twenty-Sixth Annual {ACM-SIAM} Symposium
  on Discrete Algorithms, {SODA} 2015, San Diego, CA, USA, January 4-6, 2015},
  pages 1087--1101, 2015.

\bibitem{DBLP:journals/algorithmica/EbenlendrKS14}
T.~Ebenlendr, M.~Krc{\'{a}}l, and J.~Sgall.
\newblock Graph balancing: {A} special case of scheduling unrelated parallel
  machines.
\newblock {\em Algorithmica}, 68(1):62--80, 2014.

\bibitem{DBLP:conf/innovations/FeigeJ15}
U.~Feige and S.~Jozeph.
\newblock Separation between estimation and approximation.
\newblock In {\em Proceedings of the 2015 Conference on Innovations in
  Theoretical Computer Science, {ITCS} 2015, Rehovot, Israel, January 11-13,
  2015}, pages 271--276, 2015.

\bibitem{lightgb}
C.~Huang and S.~Ott.
\newblock A combinatorial approximation algorithm for graph balancing with
  light hyper edges.
\newblock In {\em 24th Annual European Symposium on Algorithms, {ESA} 2016,
  August 22-24, 2016, Aarhus, Denmark}, pages 49:1--49:15, 2016.

\bibitem{DBLP:conf/swat/JansenLM16}
K.~Jansen, K.~Land, and M.~Maack.
\newblock Estimating the makespan of the two-valued restricted assignment
  problem.
\newblock In {\em 15th Scandinavian Symposium and Workshops on Algorithm
  Theory, {SWAT} 2016, June 22-24, 2016, Reykjavik, Iceland}, pages
  24:1--24:13, 2016.

\bibitem{Lenstra:1990:AAS:81018.81019}
J.~K. Lenstra, D.~B. Shmoys, and E.~Tardos.
\newblock Approximation algorithms for scheduling unrelated parallel machines.
\newblock {\em Mathematical Programming}, 46(3):259--271, 1990.

\bibitem{DBLP:journals/talg/PolacekS16}
L.~Pol{\'{a}}cek and O.~Svensson.
\newblock Quasi-polynomial local search for restricted max-min fair allocation.
\newblock {\em {ACM} Transactions on Algorithms}, 12(2):13, 2016.

\bibitem{schuurman1999polynomial}
P.~Schuurman and G.~J. Woeginger.
\newblock Polynomial time approximation algorithms for machine scheduling: ten
  open problems.
\newblock {\em Journal of Scheduling}, 2(5):203--213, 1999.

\bibitem{DBLP:journals/siamcomp/Svensson12}
O.~Svensson.
\newblock Santa claus schedules jobs on unrelated machines.
\newblock {\em {SIAM} Journal on Computing}, 41(5):1318--1341, 2012.

\bibitem{DBLP:journals/scheduling/VerschaeW14}
J.~Verschae and A.~Wiese.
\newblock On the configuration-lp for scheduling on unrelated machines.
\newblock {\em Journal of Scheduling}, 17(4):371--383, 2014.

\bibitem{DBLP:books/daglib/0030297}
D.~P. Williamson and D.~B. Shmoys.
\newblock {\em The Design of Approximation Algorithms}.
\newblock Cambridge University Press, 2011.

\end{thebibliography}
\end{document}